\documentclass{article}
\usepackage{mystyle2}
\usepackage{todonotes}
\usepackage{url}

\allowdisplaybreaks

\title{A Simple Optimal Contention Resolution Scheme for Uniform Matroids%
\thanks{This project received funding from the European Research Council (ERC) under the European Union's Horizon 2020 research and innovation programme (grant agreement No 817750).
\includegraphics[height=4mm]{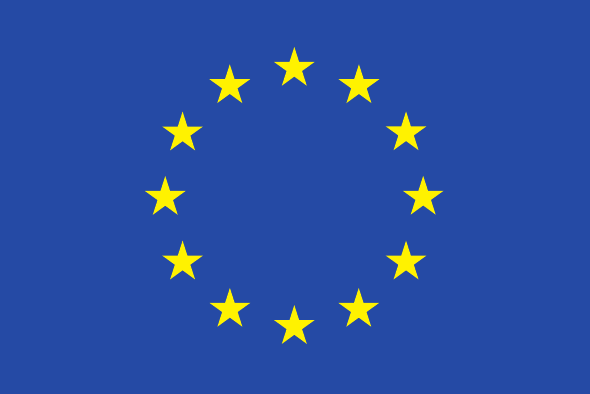} \includegraphics[height=4mm]{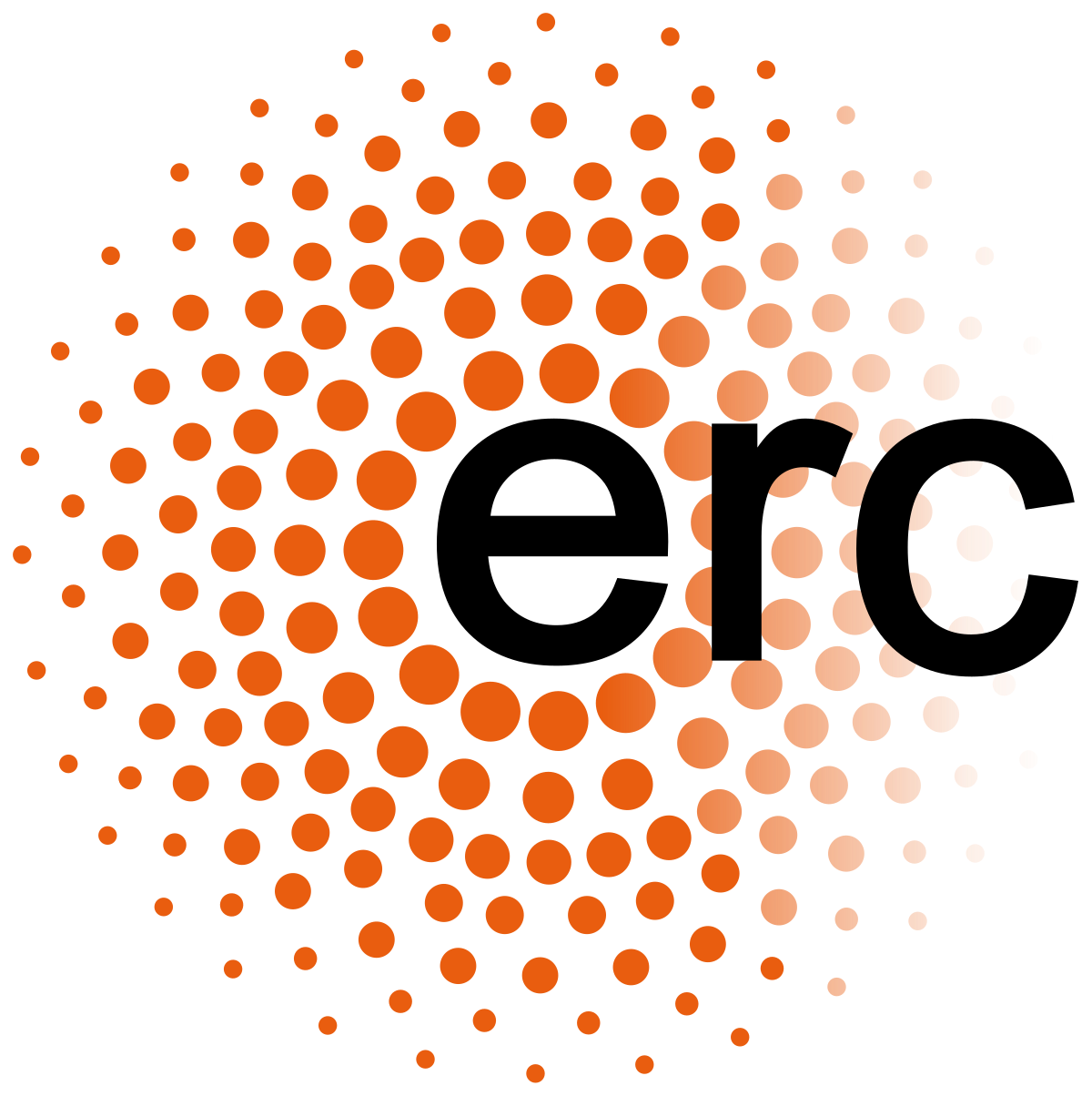}
}%
}

\author{
    Danish Kashaev \footnote{\texttt{Danish.Kashaev@cwi.nl}, ETH Z\"{u}rich, Switzerland.}
   \and
    Richard Santiago \footnote{\texttt{rtorres@ethz.ch}, ETH Z\"{u}rich, Switzerland.}}

\date{\vspace{-3ex}}

\begin{document}
\maketitle

\begin{abstract}
Contention resolution schemes (or CR schemes), introduced by Chekuri, Vondrak and Zenklusen, are
a class of randomized rounding algorithms for converting a
fractional solution to a relaxation for a down-closed constraint
family into an integer solution. A CR scheme takes a fractional point $x$ in a relaxation polytope, rounds each coordinate $x_i$ independently to get a possibly non-feasible set, and then drops some elements in order to satisfy the constraints. Intuitively, a contention resolution scheme is $c$-balanced if every element $i$ is selected with probability at least $c \cdot x_i$.

It is known that general matroids admit a $(1-1/e)$-balanced CR scheme, and that this is (asymptotically) optimal. This is in particular true for the special case of uniform matroids of rank one. 
In this work, we provide a simple and explicit monotone CR scheme for uniform matroids of rank $k$ on $n$ elements with a balancedness of $1 - \binom{n}{k}\:\left(1-\frac{k}{n}\right)^{n+1-k}\:\left(\frac{k}{n}\right)^k$, and show that this is optimal. As $n$ grows, this expression converges from above to $1 - e^{-k}k^k/k!$. While this asymptotic bound can be obtained by combining previously known results, these require defining an exponential-sized linear program, as well as using random sampling and the ellipsoid algorithm. Our procedure, on the other hand, has the advantage of being simple and explicit. This scheme extends naturally into an optimal CR scheme for partition matroids.
\end{abstract}

\section{Introduction}

Contention resolution schemes were introduced by Chekuri, Vondrak, and Zenklusen \cite{crs} 
as a tool for submodular maximization under various types of constraints. A set function $f:2^N \to \R$ is submodular if for any two sets $A\subseteq B \subseteq N$ and any element $v \notin B$, the corresponding marginal gains satisfy $f(A \cup \{v\}) -f(A) \geq f(B \cup \{v\}) -f(B)$. Submodular functions are a classical object in combinatorial optimization and operations research \cite{lovasz1983submodular}. A family of subsets $\mathcal{I} \subseteq 2^N$ is called an independence family if $\emptyset \in \mathcal{I}$ and $A\subseteq B \in \mathcal{I}$ implies $A \in \mathcal{I}$.
Given a finite ground set $N$, an independence family $\mathcal{I} \subseteq 2^N$, and a submodular set function $f:2^N \mapsto \mathbb{R}$, the problem consists of (approximately) solving $\max_{S\in \mathcal{I}}f(S)$. 

A successful technique to tackle this problem in recent years has been the relaxation and rounding approach. It consists of first relaxing the discrete problem into a continuous version $\max_{x \in P_{\I}}F(x)$, where $F:[0,1]^N \mapsto \mathbb{R}$ is a suitable continuous extension of $f$, and $P_{\I}$ is a relaxation polytope of the independence family $\mathcal{I}$
\footnote{Given $\mathcal{I} \subseteq 2^N$, let $\text{conv}(\mathcal{I})$ denote the convex hull of the set $\{\mathbf{1}_S : S \in \I\}$. Then $P \subseteq [0,1]^N$ is called a relaxation polytope of $\mathcal{I}$ if $P \cap \{0,1\}^N = \text{conv}(\mathcal{I}) \cap \{0,1\}^N$. That is, if $P$ and $\text{conv}(\mathcal{I})$ have the same set of integer points.}.
The first step of the relaxation and rounding approach then approximately solves $\max_{x \in P_{\I}}F(x)$ to obtain a fractional point $x\in P_{\mathcal{I}}$.

 In order to get a feasible solution to the original problem, we then need to round this fractional point into an integral and feasible (i.e., independent) one while keeping the objective value as high as possible. Contention resolution schemes are a powerful tool to tackle this problem, and have found other applications outside of submodular maximization \cite{adamczyk2018random,ocrs,optimal_ocrs}.

At the high level, given a fractional point $x$, the procedure first generates a random set $R(x)$ by independently including each element $i$ with probability $x_i$. Since $R(x)$ might not necessarily belong to $\mathcal{I}$, the contention resolution scheme then removes some elements from it in order to get an independent set. 
We denote the support of a point $x$ by $\text{supp}(x):= \{i \in N \mid x_i > 0\}$. A CR scheme is then formally defined as follows.

\begin{definition}[CR scheme]\label{def:crs}
	$\pi = (\pi_x)_{x\in P_{\I}}$ is a $c$-balanced \emph{contention resolution scheme} for the polytope $P_{\I}$ if for every $x \in P_{\I}$, $\pi_x$ is an algorithm that takes as input a set $A\subseteq \text{supp}(x)$ and outputs an independent set $\pi_x(A) \in \I$ contained in $A$ such that
	\begin{equation*}
		\mathbb{P}\Big[i \in \pi_x(R(x)) \mid i \in R(x)\Big] \geq c \quad \forall i \in \text{supp}(x).
	\end{equation*}
	Moreover, a contention resolution scheme is \emph{monotone} if for any $x \in P_{\I}$:
	\begin{equation*}
		\mathbb{P}[i \in \pi_x(A)] \geq \mathbb{P}[i \in \pi_x(B)] \quad \text{for any } i \in A \subseteq B \subseteq \text{supp}(x).
	\end{equation*}
\end{definition}

Notice that a contention resolution scheme can have a different algorithm for each $x \in P_{\mathcal{I}}$.
A $c$-balanced contention resolution scheme ensures that every element in the random set $R(x)$ is kept with probability at least $c$. The goal when designing CR schemes is thus to maximize such $c$, known as balancedness.

Moreover, monotonicity is a desirable property for a $c$-balanced CR scheme to have, since one can then get approximation guarantees for the constrained submodular maximization problem via the relaxation and rounding approach (see \cite{crs} for more details).

A closely related notion to contention resolution schemes is the \emph{correlation gap}, originally introduced in \cite{agrawal2012price} and extended to constraint families $\mathcal{I} \subset 2^N$ in \cite{crs}. Formally, the correlation gap of a family $\I \subseteq 2^N$ is given by
\[
\kappa(\I) = \inf_{x \in \mathtt{conv}(\I), y\geq 0} \frac{\E \left[\max_{S\subseteq R(x), S \in \I} \sum_{i \in S} y_i\right]}{\sum_{i \in N} x_i y_i}, 
\]
where $\text{conv}(\mathcal{I})$ denotes the convex hull of the set $\{\mathbf{1}_S : S \in \I\}$. This definition is the one given in \cite{crs}, while the original definition from \cite{agrawal2012price} uses the inverse ratio.
The connection between CR schemes and the correlation gap is then summarized by the following result of \cite[Theorem 4.6]{crs}: the correlation gap of $\I$ is equal to the maximum $c$ such that $\I$ admits a $c$-balanced (but not necessarily monotone) CR scheme.

By presenting a variety of CR schemes for different constraints, the work in \cite{crs} gives improved approximation algorithms for linear and submodular maximization problems under matroid, knapsack, matchoid constraints, as well as their intersections.\footnote{We postpone the formal definitions and basic matroid background to Section~\ref{sec:matroids}.}
CR schemes have also been studied for other types of independence families  \cite{crs_bipartite_matchings,correlation_gap_matching}, or by having the elements of the random set $R(x)$ arrive in an online fashion  \cite{bayesian_combinatorial_actions,ocrs,optimal_ocrs,adamczyk2018random,livanos2021simple, arnosti2021tight, jiang2022tight}. 
In this work, we restrict our attention to matroid constraints and the offline setting (i.e., we know the full set $R(x)$ in advance).

A monotone CR scheme with a balancedness of $1 - (1-1/n)^n$ for the uniform matroid of rank one is given in \cite{approximation_allocation,submodular_welfare}, where $n$ denotes the size of the ground set. It is also shown that this is optimal. That is, there is no $c$-balanced CR scheme for the uniform matroid of rank one with $c > 1 - (1-1/n)^n$. The work of \cite{crs} extends this result by proving the existence of a monotone $1 - (1-1/n)^n$-balanced CR scheme for any matroid. This requires defining an exponential-sized linear program and using its dual. The existence argument can then be turned into an efficient procedure by using random sampling and the ellipsoid algorithm to construct an efficient CR scheme with a balancedness of $1 - (1-1/n)^n - \epsilon$, running in time polynomial in the input size and $1/\epsilon$ for any fixed $\epsilon > 0$. Since $1-(1-1/n)^n$ converges to $1 - 1/e$, this corresponds to an efficient asymptotically optimal CR scheme with a balancedness of $1-1/e$. 

For the uniform matroid of rank $k$ (i.e., cardinality constraints), the work of \cite{yan2011mechanism} establishes a correlation gap of $1 - e^{-k} k^k/k!$. Combining this with a reduction from \cite{crs}
proves the existence of a $(1 - e^{-k} k^k/k!)$-balanced CR scheme, and the asymptotic optimality of this bound. The existence of such a scheme can also be obtained by combining the same reduction from \cite{crs} and a result of \cite{barman2021tight}. However, the main drawback of these approaches is their lack of simplicity. 

In the setting where the elements of $R(x)$ arrive in an online adversarial fashion, the work of  \cite{bayesian_combinatorial_actions} gives a procedure with a balancedness of $1-1/\sqrt{k+3}$ for uniform matroids of rank $k$.
It remained unknown whether this bound was tight, until the recent work of \cite{jiang2022tight} settled this question.
They show that the optimal balancedness in this setting is strictly better than $1-1/\sqrt{k+3}$, and strictly worse than $1 - e^{-k} k^k/k!$. In contrast, for the case where the elements arrive in a random fashion, it has recently been shown that the optimal balancedness is $1 - e^{-k} k^k/k!$  \cite{arnosti2021tight}.

\subsection{Our contributions}

Our main result is to provide a simple, explicit, and optimal monotone CR scheme for the uniform matroid of rank $k$ on $n$ elements, with a balancedness of $c(k,n):=1 - \binom{n}{k}\left(1-\frac{k}{n}\right)^{n+1-k}\left(\frac{k}{n}\right)^k$. This result is encapsulated in Theorem~\ref{thm_optimal_crs_ukn} (balancedness), Theorem~\ref{thm_optimality} (optimality), and Theorem~\ref{thm_monotonicity_cr_scheme_ukn} (monotonicity). This generalizes the balancedness factor of $1 - (1-1/n)^n$  given in \cite{approximation_allocation,submodular_welfare} for the rank one (i.e., $k=1$) case. Moreover, for a fixed value of $k$, we have that $c(k,n)$ converges from above to $1 - e^{-k} k^k/k!$. While it is possible to prove the existence of a ($1 - e^{-k} k^k/k!$)-balanced CR scheme by combining results from \cite{crs,yan2011mechanism}, these require defining an exponential-sized linear program and using its dual. In addition, to turn this existence proof into an actual algorithm, one needs to use random sampling and the ellipsoid method. The advantage of our CR scheme is thus that it is a very simple and explicit procedure. Moreover, our balancedness is an explicit formula which also depends on $n$ (the number of elements) in addition to $k$, and $c(k,n) > 1 - e^{-k} k^k/k!$ for every fixed $n$.
We also discuss how the above CR scheme for uniform matroids naturally generalizes to partition matroids. 

\subsection{Preliminaries on matroids}
\label{sec:matroids}

This section provides a brief background on matroids. A \emph{matroid} $\mathcal{M}$ is a pair $(N,\I)$ consisting of a ground set $N$ and a non-empty family of \emph{independent sets} $\I \subseteq 2^N$ which satisfy:
\begin{itemize}
	\item If $A \in \I$ and $B \subseteq A$, then $B \in \I$.
	\item If $A \in \I$ and $B \in \I$ with $|A| > |B|$, then $\exists \: i\in A\backslash B$ such that $B \cup \{i\} \in \I$.
\end{itemize}

\noindent Given a matroid  $\mathcal{M}=(N,\I)$ its \emph{rank function} $r:2^N \to \mathbb{R}_{\geq 0}$ is defined as $r(A)=\max\{|S|: S \subseteq A,\; S\in \mathcal{I}\}$. Its \emph{matroid polytope} is given by
$P_{\I} := \text{conv} (\{\mathbf{1}_S : S \in \I\}) = \{x \in \mathbb{R}^N_{\geq 0}: x(A) \leq r(A), \; \forall A \subseteq N\}$,
where $x(A) := \sum_{i \in A}x_i$. Note that this implies $x_i \in [0,1]$ for every $i \in N$.

The next two classes of matroids are of special interest for this work.

\begin{example}[Uniform matroid]
	\label{ex:uniform}
	The uniform matroid of rank $k$ on $n$ elements $U^k_n:= (N, \I)$ is the matroid whose independent sets are all the subsets of the ground set of cardinality at most $k$. That is, $\I := \{A \subseteq N: |A| \leq k\}$. Its matroid polytope is $P_{\I} = \{x \in [0,1]^N: x(N) \leq k\}$.
\end{example}

\begin{example}[Partition matroid]
	\label{ex:partition}
	Partition matroids are a generalization of uniform matroids, where the ground set is partitioned into $k$ blocks: $N = D_1 \uplus \dots \uplus D_k$ and each block $D_i$ has a certain capacity $d_i \in \mathbb{Z}_{\geq 0}$. The independent sets are then defined to be $\I:= \{A \subseteq N: |A \cap D_i| \leq d_i, \; \forall i \in \{1,\dots,k\}\}$. 
	The matroid polytope in this case is $P_{\I} = \{x \in [0,1]^N: x(D_i) \leq d_i, \; \forall i\in \{1,\dots,k\}\}$.
	The uniform matroid $U^k_n$ is simply a partition matroid with one block $N$ and one capacity $k$. Moreover, the restriction of a partition matroid to each block $D_i$ is a uniform matroid of rank $d_i$ on the ground set $D_i$.
\end{example}

\section{An optimal monotone contention resolution scheme for uniform matroids}
\label{sec:uniform}


We assume throughout this whole section that $n \geq 2$ and that $k \in \{1,\dots,n-1\}$. 
We denote by $P_{\I}$ the matroid polytope of $U_n^k$.
For any point $x \in P_{\I}$, let $R(x)$ be the random set satisfying $\mathbb{P}[i \in R(x)] = x_i$ independently for each coordinate. 
If the size of $R(x)$ is at most $k$, then $R(x)$ is already an independent set and the CR scheme returns it. If however $|R(x)| > k$, then the CR scheme returns a random subset of $k$ elements by making the probabilities of each subset of $k$ elements depend linearly on the coordinates of the original point $x \in P_{\I}$.
More precisely, given an arbitrary $x \in P_{\I}$, for any set $A \subseteq\text{supp}(x)$ with $|A| > k$ and any subset $B \subseteq A$ of size $k$, we define
\begin{equation}
	\label{q_A_B}
	q_A(B):= \frac{1}{\mbinom{|A|}{|B|}} \Big(1 + \: \bar{x}(A\setminus B) - \bar{x}(B) \Big),
\end{equation}
where we use the following notation:
$\bar{x}(A):= \frac{1}{|A|}x(A) = \frac{1}{|A|}\sum_{i \in A}x_i$.
We then define a randomized CR scheme $\pi$ for $U^k_n$ as follows.

\begin{algo}[CR scheme $\pi$ for $U^k_n$]
	\label{algo_optimal_crs_ukn}
	
	We are given a point $x\in P_{\I}$ and a set $A \subseteq \text{supp}(x)$. \begin{itemize}
		\item If $|A| \leq k$, then $\pi_x(A) = A$.
		\item If $|A| > k$, then for every $B \subseteq A$ with $|B| = k, \; \pi_x(A) = B$ with probability $q_A(B)$.
	\end{itemize}
\end{algo}
\noindent The above procedure can be implemented in $O(n^k)$ time in the worst case, and hence gives a polynomial time algorithm for constant values of $k$. In Section~\ref{sec:efficient} we discuss an alternative viewpoint of the scheme, which yields a more efficient implementation and a polynomial time algorithm even for non-constant $k$.

We next show that the above CR scheme is well-defined, i.e., that $q_A$ is a valid probability distribution.

\begin{lemma}
	The above procedure $\pi$ is a well-defined CR scheme. That is, $\forall x\in P_{\I}$ and $A \subseteq \text{\emph{supp}}(x)$, we have $q_A(B) \geq 0$ and $\sum_{B \subseteq A, |B|=k} q_A(B) = 1$.
\end{lemma}

\begin{proof}
	Since $\bar{x}(A \setminus B) \in [0,1]$ and $\bar{x}(B) \in [0,1]$, it directly follows from the definition \eqref{q_A_B} that $q_A(B) \geq 0$.
	In order to prove the second claim, we need the equality
	\begin{equation}
		\label{eq_sum_x_B}
		\begin{split}
			\sum_{B \subseteq A, |B|=k}x(B) = \binom{|A|-1}{k-1} \; x(A)
		\end{split}
	\end{equation}
	that we derive the following way:
	\begin{align*}
		\sum_{B \subseteq A, |B|=k}x(B) &= \sum_{B \subseteq A,|B|=k} \; \sum_{i \in A}x_i \mathbf{1}_{\{i \in B\}} = \sum_{i \in A} x_i \sum_{B \subseteq A, |B|=k} \mathbf{1}_{\{i \in B\}} \\&= \sum_{i \in A} x_i \; |\{B \subseteq A \mid |B|=k, i \in B\}| = \binom{|A|-1}{k-1} \; x(A).
	\end{align*}
	
	\noindent Hence,
	\begin{align*}
		\sum_{B \subseteq A, |B|=k} q_A(B) &= \sum_{B \subseteq A, |B|=k} \frac{1}{\mbinom{|A|}{k}}\left(1 + \frac{x(A\setminus B)}{|A| - k} \:  - \frac{x(B)}{k} \right)\\
		&= 1 + \frac{1}{\mbinom{|A|}{k}}\sum_{B \subseteq A, |B| = k} \left(\frac{x(A)}{|A|-k} - \frac{x(B)}{|A|-k}-\frac{x(B)}{k} \right)\\
		&= 1 + \frac{1}{\mbinom{|A|}{k} \: \Big(|A|-k\Big)}\sum_{B \subseteq A, |B| = k} \left(x(A) - \frac{|A|}{k}x(B) \right)\\
		&= 1 + \frac{1}{\mbinom{|A|}{k} \: \Big(|A|-k\Big)} \left( \mbinom{|A|}{k}x(A) - \mbinom{|A|}{k}x(A) \right) = 1. \qedhere
	\end{align*}
\end{proof}

We now state our main result.

\begin{theorem}
	\label{thm_optimal_crs_ukn}
	Algorithm \ref{algo_optimal_crs_ukn} is a $c$-balanced CR scheme for the uniform matroid of rank $k$ on $n$ elements, where $c = 1 - \binom{n}{k}\:\left(1-\frac{k}{n}\right)^{n+1-k}\:\left(\frac{k}{n}\right)^k$.
\end{theorem}
\noindent Since we use the above expression often throughout this section, we denote it by
\begin{equation*}
	c(k,n):= 1 - \binom{n}{k}\:\left(1-\frac{k}{n}\right)^{n+1-k}\:\left(\frac{k}{n}\right)^k.
\end{equation*}
\noindent We note that setting $k = 1$ gives $c(1,n) = 1 - (1-1/n)^n$, which matches the optimal balancedness for $U^1_n$ provided in \cite{approximation_allocation, submodular_welfare}. 
This converges to $1-1/e$ when $n$ gets large.

\begin{proposition}
	For a fixed $k$, the limit of $c(k,n)$ as $n$ tends to infinity is
	\[\lim_{n \to \infty} c(k,n) = 1 - e^{-k} \: \frac{k^k}{k!}.\]
	Moreover, $c(k,n)$ is monotonically decreasing with $n$ for a fixed $k$.
\end{proposition}
\begin{proof}
	We use Stirling's approximation, which states that:
	\begin{equation}
		\label{eq_stirling_approx}
		n! \sim \sqrt{2\pi n} \left(\frac{n}{e}\right)^n.
	\end{equation}
	This means that these two quantities are asymptotic, i.e., their ratio tends to 1 if we tend $n$ to infinity. By \eqref{eq_stirling_approx}, we get
	\begin{equation*}
		\frac{n!}{(n-k)!} \sim \sqrt{2 \pi n}\left(\frac{n}{e}\right)^n \frac{1}{\sqrt{2 \pi (n-k)}}\left(\frac{e}{n-k}\right)^{n-k} =  e^{-k}  \frac{n^n}{(n-k)^{n-k}}\sqrt{\frac{n}{n-k}}.
	\end{equation*}
	Using the above expression leads to the desired result:
	\begin{align*}
		1 - c(k,n) &= \binom{n}{k}\:\left(1-\frac{k}{n}\right)^{n+1-k}\:\left(\frac{k}{n}\right)^k 
		=\frac{k^k}{k!} \; \frac{n!}{(n-k)!} \; \frac{(n-k)^{n+1-k}}{n^{n+1}} \\
		&\sim e^{-k}\; \frac{k^k}{k!} \; \frac{n-k}{n} \sqrt{\frac{n}{n-k}} 
		= e^{-k}\; \frac{k^k}{k!} \; \sqrt{\frac{n-k}{n}}
		\sim e^{-k}\; \frac{k^k}{k!}.
	\end{align*}
In order to prove that $c(k,n)$ is monotonically decreasing with $n$, we now show that
\begin{equation}
\label{eq_ckn_decreasing_n}
\frac{1-c(k,n+1)}{1-c(k,n)} > 1.
\end{equation}
By expanding this expression, we get
\begin{align*}
\frac{1-c(k,n+1)}{1-c(k,n)} &= \frac{\binom{n+1}{k}}{\binom{n}{k}}\; \frac{\big(1 - k/(n+1)\big)^{n+2-k}}{\big(1 - k/n\big)^{n+1-k}} \; \frac{(k/(n+1))^{k}}{(k/n)^k} \\ &= \left(\frac{n}{n+1}\right)^{n+1} \left(\frac{n+1-k}{n-k}\right)^{n+1-k} = \quad \frac{g(n+1-k)}{g(n+1)},
\end{align*}
where the function $g(x)$ is defined as $g(x) := (x/(x-1))^x$. In order for \eqref{eq_ckn_decreasing_n} to hold, it now suffices to show that $g(x)$ is monotonically decreasing for $x > 1$. We do that by showing that the derivative of the logarithm is strictly negative. 
\begin{align*}
\frac{d}{dx} \log(g(x)) &= \frac{d}{dx} x \log\left(\frac{x}{x-1}\right) = \log\left(\frac{x}{x-1}\right) + x \left(\frac{1}{x} - \frac{1}{x-1}\right) = \log\left(\frac{x}{x-1}\right) - \frac{1}{x-1}\\
  		&= \log\left(1 + \frac{1}{x-1}\right)-\frac{1}{x-1} < 0,
\end{align*}
where the last inequality follows from the fact that $\log(1+y) < y$ for any $y > 0$. This shows that \eqref{eq_ckn_decreasing_n} holds and hence that $c(k,n)$ is monotonically decreasing with $n$. \qedhere
\end{proof}

\subsection{Outline of the proof of Theorem \ref{thm_optimal_crs_ukn}}
\label{subsection_outline_proof}

Throughout this whole section on uniform matroids, we fix an arbitrary element $e \in N$. In order to prove Theorem \ref{thm_optimal_crs_ukn}, we need to show that for every $x \in P_{\I}$ with $x_e > 0$ we have $\mathbb{P}[e \in \pi_x(R(x)) \mid e \in R(x)] \geq c(k,n)$.
This is equivalent to showing that for every $x \in P_{\I}$ with $x_e > 0$ we have
\begin{equation}
	\label{eq_proba_e_not_in_pi_x_R_x}
	\mathbb{P}[e \notin \pi_x(R(x)) \mid e \in R(x)] \leq 1 - c(k,n).
\end{equation}

We now introduce some definitions and notation that will be needed.
For any $B \subseteq A \subseteq N$, let
$p_A(B) := \mathbb{P}[R_A(x) = B] = \prod_{i \in B}x_i \prod_{i \in A \setminus B} (1-x_i)$, where $R_A(x)$ is the random set obtained by rounding each coordinate of $x|_A$ in the reduced ground set $A$ to one independently with probability $x_i$.
That is, $R_A(x)=R(x)\cap A$.
Note that $p_N(B) = \mathbb{P}[R(x) = B]$. We do not write the dependence on $x\in P_{\I}$ for simplicity of notation.
We mainly work on the set $N \setminus \{e\}$. For this reason, we define $S := N \setminus \{e\}$.
Note that $|S| = n-1$; we use this often in our arguments.

With the above notation we can rewrite the probability in \eqref{eq_proba_e_not_in_pi_x_R_x} in a more convenient form. For any $x \in P_{\I}$ satisfying $x_e > 0$, we get
\begin{align}
	\begin{split}
		\mathbb{P}\Big[e \notin \pi_x(R(x)) \mid e \in R(x)\Big] &= \sum_{A \subseteq S}\mathbb{P}[e \notin \pi_x(R(x)) \mid R_S(x) = A, e \in R(x)]\;\mathbb{P}[R_S(x)=A \mid e \in R(x)]\\
		&= \sum_{A \subseteq S, |A| \geq k}\mathbb{P}\Big[e \notin \pi_x(R(x)) \mid R(x)=A \cup e\Big] \; p_S(A) \nonumber\\
		&= \sum_{A \subseteq S, |A| \geq k} p_S(A) \sum_{B \subseteq A, |B|=k} q_{A \cup e}(B). \nonumber\\
	\end{split}
\end{align}
The obtained expression is a multivariable function of the variables $x_1,\dots,x_n$, since $p_S(A)$ and $q_{A \cup e}(B)$ depend on those variables as well. We denote it as follows.
\begin{equation}
	\label{eq_def_G_e}
	G(x) := \sum_{A \subseteq S, |A| \geq k} p_S(A) \sum_{B \subseteq A, |B|=k} q_{A \cup e}(B).
\end{equation}
\noindent One then has that for proving Theorem \ref{thm_optimal_crs_ukn} it suffices to show the following.
\begin{theorem}
	\label{thm_max_crs_ukn}
	Let $G(x)$ and $c(k,n)$ be as defined above. Then $\max_{x \in P_{\I}} G(x) = 1 - c(k,n)$. Moreover, the maximum is attained at the point $(x_1,\dots,x_n) = (k/n,\dots,k/n) \in P_{\I}.$
\end{theorem}
\noindent Indeed, Theorem \ref{thm_max_crs_ukn} implies that for every $x \in P_{\I}$ we have $G(x) \leq 1-c(k,n)$,
with equality holding if $x = (k/n,\dots,k/n)$. In particular, for any $x \in P_{\I}$ satisfying $x_e > 0$, we get
\[
G(x) = \mathbb{P}\Big[e \notin \pi_x(R(x)) \mid e \in R(x)\Big] \leq 1 - c(k,n),
\]
which proves Theorem \ref{thm_optimal_crs_ukn} by 
\eqref{eq_proba_e_not_in_pi_x_R_x}.

Notice that for the conditional probability to be well defined, we need the assumption that $x_e > 0$. However, in our case $G(x)$ is simply a multivariable polynomial function of the $n$ variables $x_1,\dots,x_n$ and is thus also defined when $x_e=0$. We may therefore forget the conditional probability and simply treat Theorem \ref{thm_max_crs_ukn} as a multivariable maximization problem over a bounded domain. We now state the outline of the proof for Theorem \ref{thm_max_crs_ukn}.

We first maximize $G(x)$ over the variable $x_e$, and get an expression depending only on the $x$-variables in $S$. This is done in Section \ref{subsection_maximizing_variable_x_e}.
We then maximize the above expression over the unit hypercube $[0,1]^S$ (see Section \ref{subsection_maximizing_h_S_k}).
Finally, we combine the above two results to show that the maximum in Theorem \ref{thm_max_crs_ukn} is attained at the point $x_i=k/n$ for every $i \in N$; this is done in Section \ref{subsection_proof_of_main_theorem_uniform}.

\subsection{Maximizing over the variable $x_e$}
\label{subsection_maximizing_variable_x_e}

The matroid polytope of $U^k_n$ is given by 
$P_{\I} = \{x \in [0,1]^N: x(N) \leq k\}$.
We define a new polytope by removing the constraint $x_e \leq 1$ from $P_{\I}$:
\begin{equation*}
	\widetilde{P}_{\I}:= \{x \in \mathbb{R}^N_{\geq 0}: x(N) \leq k \text{ and } x_i \leq 1 \;\; \forall i \in S\}.
\end{equation*}
Clearly, $P_{\I} \subseteq \widetilde{P}_{\I}$.

We now present the main result of this section, where we consider the maximization problem $\max\{G(x) \mid x \in \widetilde{P}_{\I}\}$ and maximize $G(x)$ over the variable $x_e$ while keeping all the other variables ($x_i$ for every $i \in S$) fixed to get an expression depending only on the $x$-variables in $S$.

\begin{lemma}
	\label{technical_lemma}
	For every $x \in \widetilde{P}_{\I}$, 
	\begin{equation}
		\label{technical_equation_bound}
		G(x) \leq \sum_{A \subseteq S, |A|=k} p_S(A) \Big(1-\bar{x}(A)\Big).
	\end{equation}
	Moreover, equality holds when $x_e = k - x(S)$.
\end{lemma}

\begin{proof}
	\begin{align}
		\label{equation_lemma_part_one}
		G(x) &= \sum_{A \subseteq S, |A| \geq k} p_S(A) \sum_{B \subseteq A, |B|=k} q_{A \cup e}(B) \nonumber\\
		&= \sum_{A \subseteq S, |A| \geq k} p_S(A) \sum_{B \subseteq A, |B| = k} \frac{1}{\mbinom{|A|+1}{k}} \Big(1 + \: \bar{x}\Big((A\setminus B) \cup e)\Big) - \bar{x}(B) \Big) \nonumber\\
		&= \sum_{A \subseteq S, |A| \geq k} p_S(A) \frac{1}{\mbinom{|A|+1}{k}} \sum_{B \subseteq A, |B| = k} \left(1 + \: \frac{x(A \setminus B) + x_e}{|A|-k+1} - \frac{x(B)}{k}\right).
	\end{align}
	We now maximize this expression with respect to the variable $x_e$ over $\widetilde{P}_{\I}$ while keeping all the other variables fixed. Since this is a linear function of $x_e$ and the coefficient of $x_e$ is positive, the maximal value will be $x_e = k - x(S)$ in order to satisfy the constraint $x(N) \leq k$. Note that this was the reason for the definition of $\widetilde{P}_{\I}$, since $k-x(S)$ might not necessarily be smaller than 1. We thus plug-in $x_e = k-x(S)$ in \eqref{equation_lemma_part_one} and write an inequality to emphasize that the derivation holds for any $x \in \widetilde{P}_{\I}$.
	
	\begin{align}
		\begin{split}
			\label{equation_lemma_part_two}
			\text{\eqref{equation_lemma_part_one}} &\leq \sum_{A \subseteq S, |A| \geq k} p_S(A) \frac{1}{\mbinom{|A|+1}{k}} \sum_{B \subseteq A, |B| = k} \left(1 + \: \frac{x(A \setminus B) + k - x(S)}{|A|-k+1} - \frac{x(B)}{k}\right)\\
			&=  \sum_{A \subseteq S, |A| \geq k} p_S(A) \frac{1}{\mbinom{|A|+1}{k}} \sum_{B \subseteq A, |B| = k} \left(1 + \: \frac{k-x(S \setminus A)-x(B)}{|A|-k+1} - \frac{x(B)}{k}\right)\\
			&= \sum_{A \subseteq S, |A| \geq k} p_S(A) \frac{1}{\mbinom{|A|+1}{k}} \sum_{B \subseteq A, |B| = k}  \left(\frac{|A|+1}{|A|-k+1} - \:\frac{x(S \setminus A)}{|A|-k+1} - \frac{|A|+1}{k(|A|-k+1)} x(B) \right).
		\end{split}
	\end{align}
	Notice the only part which depends on $B$ in the last summation is $x(B)$. By using Equation \eqref{eq_sum_x_B} and noticing that $\sum_{B \subseteq A, |B| =k} 1 = \mbinom{|A|}{k}$, we get
	
	\begin{align}
		\eqref{equation_lemma_part_two} &= \sum_{A \subseteq S, |A| \geq k} p_S(A) \frac{1}{\mbinom{|A|+1}{k}} \frac{1}{|A|-k+1}\left(\mbinom{|A|}{k}\big(|A|+1 - x(S \setminus A)\big)  - \frac{|A|+1}{k} \mbinom{|A|-1}{k-1}x(A)\right) \nonumber \\
		&= \sum_{A \subseteq S, |A| \geq k} p_S(A) \frac{1}{\mbinom{|A|+1}{k}} \frac{1}{|A|-k+1} \mbinom{|A|}{k}\left(|A|+1 - x(S \setminus A)  - \frac{|A|+1}{|A|} x(A)\right) \nonumber \\
		&= \sum_{A \subseteq S, |A| \geq k} \; \frac{p_S(A)}{|A|+1} \left(|A|+1 - x(S \setminus A)  - \frac{|A|+1}{|A|} x(A)\right) \nonumber \\
		&= \sum_{A \subseteq S, |A| \geq k} \; p_S(A) \left(1 - \frac{x(S \setminus A)}{|A|+1}  - \frac{x(A)}{|A|}\right). \label{equation_lemma_part_three}
	\end{align}
	\noindent Now, note that by definition of the term $p_S(A)$, we have
	\begin{equation}
		\label{vondrak_trick}
		x_i \: p_S(A) = (1-x_i) \: p_S(A \cup i) \quad \text{ for any } i \in S \setminus A.
	\end{equation}
	\noindent We compute the middle term in \eqref{equation_lemma_part_three} by plugging in \eqref{vondrak_trick} and the change of variable $B:= A \cup i$.
	
	\begin{align}
		\label{middle_term_vondrak_trick}
		\sum_{A \subseteq S, |A| \geq k} \; \frac{1}{|A|+1} &p_S(A) \: x(S \setminus A) = 
		\sum_{A \subseteq S, |A| \geq k} \; \sum_{i\in S}\; \frac{1}{|A|+1} \: x_i \: p_S(A) \: \mathbf{1}_{\{i \notin A\}} \nonumber \\
		&= \sum_{i\in S}\ \; \sum_{A \subseteq S, |A| \geq k} \; \frac{1}{|A|+1} \: (1-x_i) \: p_S(A \cup i) \: \mathbf{1}_{\{i \notin A\}} \nonumber\\
		&= \sum_{i\in S}\ \; \sum_{B \subseteq S, |B| \geq k+1} \; \frac{1}{|B|} \: (1-x_i) \: p_S(B) \: \mathbf{1}_{\{i \in B\}} \nonumber\\
		&= \sum_{B \subseteq S, |B| \geq k+1} \; \frac{1}{|B|} \: p_S(B) \sum_{i\in S}\; \mathbf{1}_{\{i \in B\}} - \sum_{B \subseteq S, |B| \geq k+1} \; \frac{1}{|B|} \: p_S(B) \sum_{i\in S}\; x_i \mathbf{1}_{\{i \in B\}} \nonumber\\
		&= \sum_{B \subseteq S, |B| \geq k+1} p_S(B) - \sum_{B \subseteq S, |B| \geq k+1}  \frac{p_S(B)}{|B|} \; x(B) \nonumber\\ 
		&= \sum_{B \subseteq S, |B| \geq k+1}  p_S(B)  \left(1-\frac{x(B)}{|B|} \right) \nonumber\\
		&= \sum_{A \subseteq S, |A| \geq k+1}  p_S(A)  \left(1-\frac{x(A)}{|A|} \right).
	\end{align}
	We finally plug-in \eqref{middle_term_vondrak_trick} into \eqref{equation_lemma_part_three} and use $\sum_{A \subseteq S, |A| \geq k} = \sum_{A \subseteq S, |A| \geq k+1} + \sum_{A \subseteq S, |A| = k}$ to get
	
	\begin{align*}
		\eqref{equation_lemma_part_three}&= \sum_{A \subseteq S, |A| = k} p_S(A)  \left(1-\frac{x(A)}{|A|} \right)
		= \sum_{A \subseteq S, |A|=k} p_S(A) \Big(1-\bar{x}(A)\Big).
	\end{align*}
	Notice that the only place where we used an inequality was from \eqref{equation_lemma_part_one} to \eqref{equation_lemma_part_two}. Hence equality holds when $x_e = k-x(S)$.
\end{proof}

\subsection{Maximizing $h_S^k: [0,1]^S \mapsto \mathbb{R}$}
\label{subsection_maximizing_h_S_k}

In this section, we turn our attention into maximizing the right-hand side expression in \eqref{technical_equation_bound} over the unit hypercube $[0,1]^S$.
In fact, we work with the following function instead:
\begin{equation*}
	h_S^k(x):= \sum_{A \subseteq S, \,|A|=k}p_S(A)(k-x(A)).
\end{equation*}
A plot of $h_S^k(x)$ for $S = \{1,2\}$ and $k=1,2$ is presented in Figure \ref{figure_hsk}. 
Note that the above function 
is simply the right hand side of \eqref{technical_equation_bound} multiplied by $k$. Hence, maximizing one or the other is equivalent. 

\begin{theorem}
	\label{thm_max_hSk}
	Let $n \geq 2$, so that $|S| = n-1 \geq 1$ and $k \in \{1,\dots,n-1\}$.
	Then the function $h_S^k(x)$
	attains its maximum over the unit hypercube $[0,1]^S$ at the point $(k/n, \dots, k/n)$ with value
	\begin{equation*}
		h_S^k\Big(k/n, \dots, k/n\Big) = k \; \binom{n}{k}\:\left(1-\frac{k}{n}\right)^{n+1-k}\:\left(\frac{k}{n}\right)^k = k\Big(1-c(k,n)\Big).
	\end{equation*}
\end{theorem}

\noindent For simplicity, we denote this maximum value by:
\begin{equation*}
	\alpha(k,n) := k \; \binom{n}{k}\:\left(1-\frac{k}{n}\right)^{n+1-k}\:\left(\frac{k}{n}\right)^k.
\end{equation*}

Notice that $h_S^0(x) = h_S^n(x) = 0$ for any $x \in [0,1]^S$. Hence Theorem~\ref{thm_max_hSk} holds for $k=0$ and $k=n$ as well. Moreover, the function $h_S^k(x)$ also satisfies an interesting duality property: $h_S^k(x) = h_S^{n-k}(1-x)$.

In order to prove Theorem \ref{thm_max_hSk}, we first show that $h$ has a unique extremum (in particular a local maximum) in the interior of $[0,1]^S$ at the point $(k/n,\dots,k/n)$ --- see Proposition \ref{prop_interior}. We then use induction on $n$ to show that any point in the boundary of $[0,1]^S$ has a lower function value than $h_S^k(k/n,\dots,k/n)$.  Since our function is continuous over a compact domain, it attains a maximum. That maximum then has to be attained at $(k/n,\dots,k/n)$ by the two arguments above. That is, the unique extremum cannot be a local minimum or a saddle point. Otherwise, since there are no more extrema in the interior and the function is continuous, the function would increase in some direction leading to a point in the boundary having higher value.
For completion, in the appendix, we present another proof showing local maximality that relies on the Hessian matrix.

\begin{proposition}
	\label{prop_interior}
	For any $k \in \{1,\dots,n-1\}, h_S^k(x)$ has a unique extremum in the interior of the unit hypercube $[0,1]^S$ at the point $(k/n,\dots,k/n)$.
\end{proposition}

\begin{figure}[t]
	\center
	\begin{subfigure}{0.4 \textwidth}
		\center
		\includegraphics[width = \textwidth]{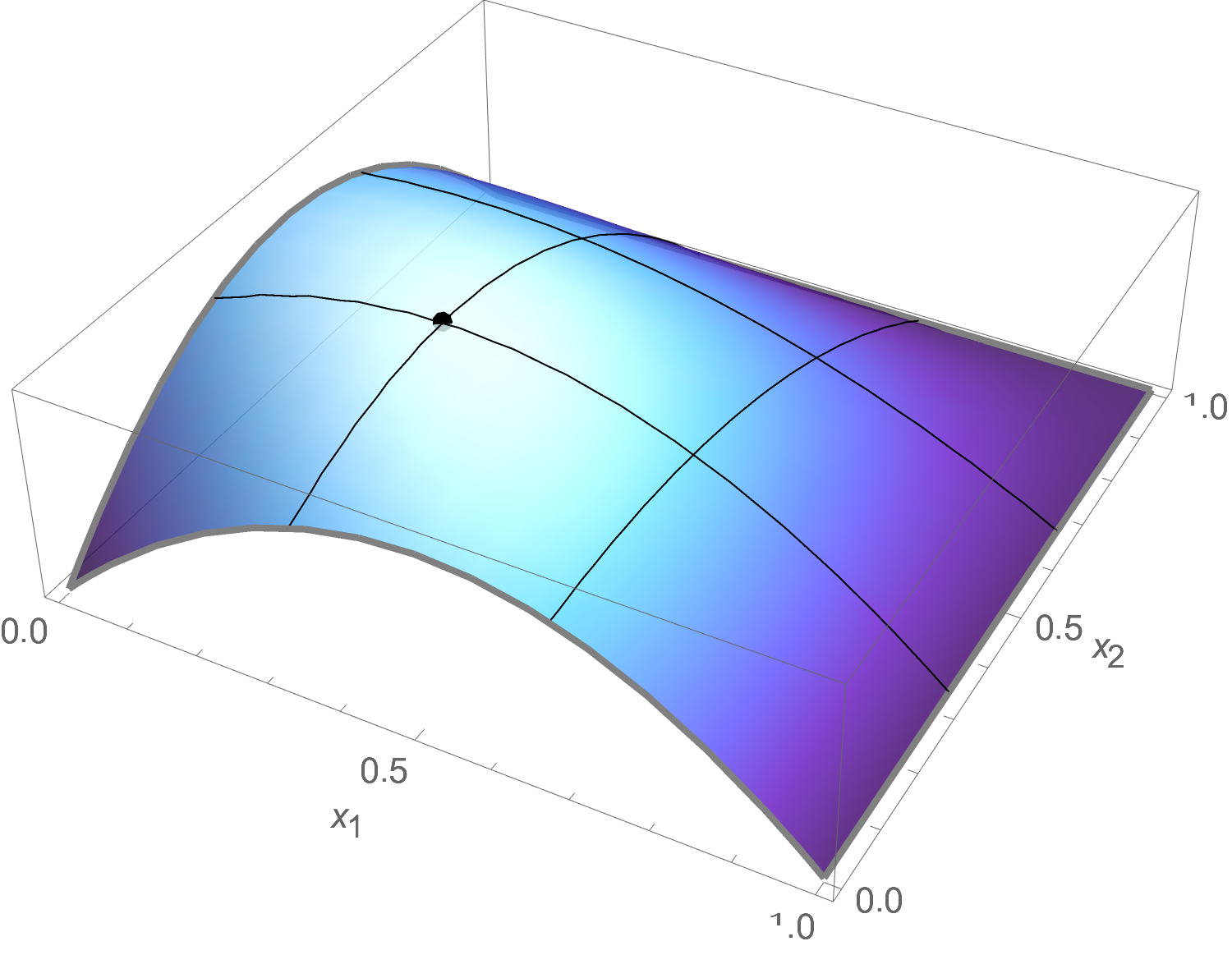}
		\caption{$S=\{1,2\}, k=1$}
		\label{figure_hsk_a}
	\end{subfigure}
	\hspace*{2cm}
	\begin{subfigure}{0.4\textwidth}
		\center
		\includegraphics[width = \textwidth]{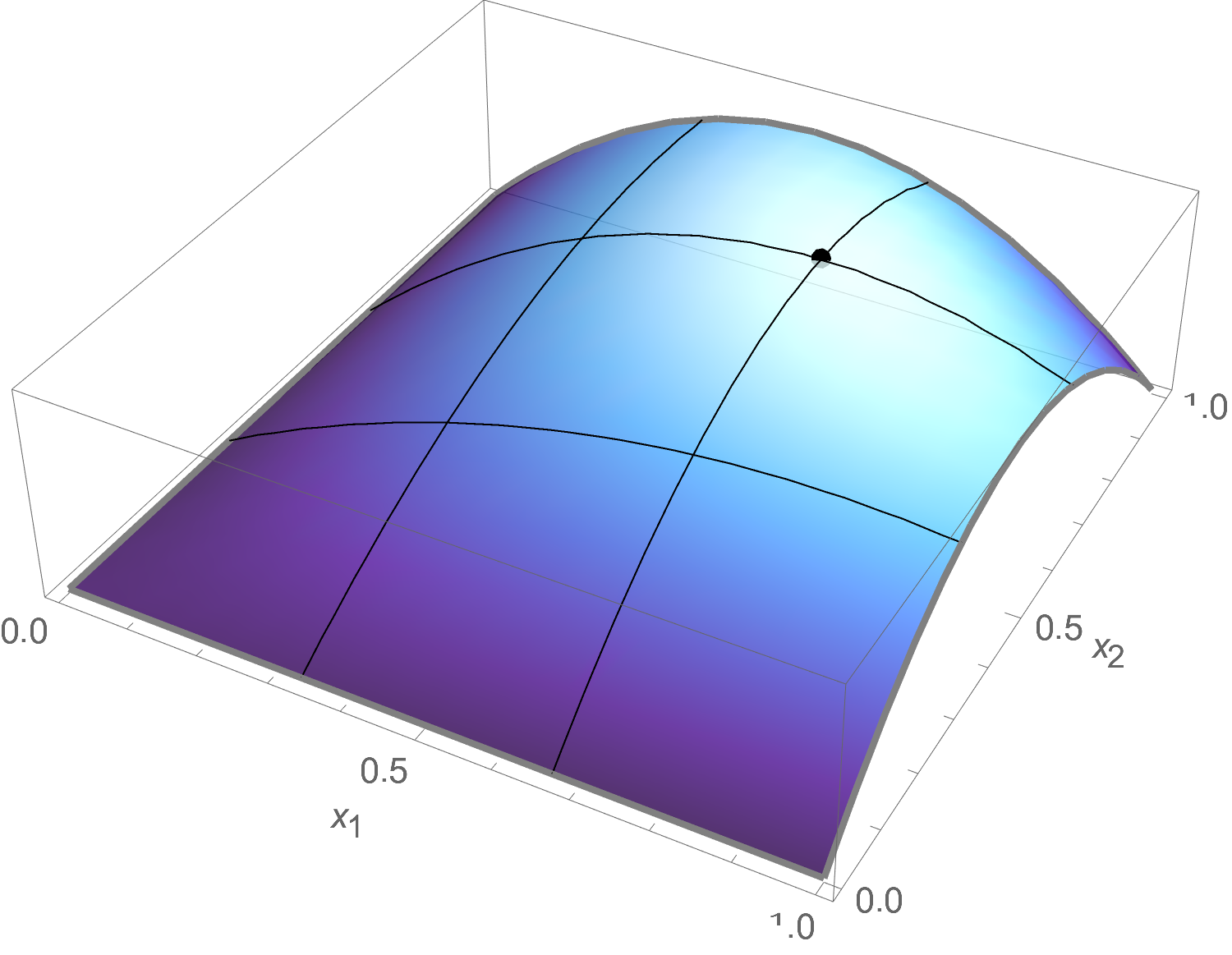}
		\caption{$S=\{1,2\}, k=2$}
	\end{subfigure}
	\caption{Plot of $h_S^k(x)$ for $S=\{1,2\}$. The maximum is attained at $x_1 = x_2 = 1/3$ in (a) and at $x_1 = x_2 = 2/3$ in (b).}
	\label{figure_hsk}
\end{figure} 

For proving Proposition \ref{prop_interior} we need the following lemma. We leave its proof to the appendix.
\begin{lemma}
	\label{lemma_recursive_h_S_k}
	The following holds for any $x \in [0,1]^S$:
	\begin{equation*}
		h_S^k(x) = \sum_{i = 0}^{k-1} Q_S^i(x) \: \big(x(S) - i\big)
	\end{equation*}
	where
	\begin{equation*}
		Q_S^k(x) := \sum_{A \subseteq S, |A| = k} p_S(A).
	\end{equation*}
\end{lemma}

The above formula actually holds for $h_A^k$ with any $A \subseteq N$. We use this in Section \ref{subsection_optimality_uniform} with $A = N$.
We are now able to prove Proposition \ref{prop_interior}.

\begin{proof}[Proof of Proposition \ref{prop_interior}]
	Let $k \in \{1,\dots,n-1\}$. To find the extrema of $h_S^k:[0,1]^S \mapsto \mathbb{R}$, we want to solve $\nabla h_S^k(x)=0$. We thus first need to compute the partial derivatives $\frac{\partial h_S^k(x)}{\partial x_i}$ for every $i \in S$.
	Note that for a set $A \subseteq S$ such that $i \in A$, we have
	\begin{align*}
		\begin{split}
			\frac{\partial}{\partial x_i}p_S(A)(k-x(A)) &= (k - x(A)) \prod_{j \in A \setminus i}x_j \prod_{j \in S \setminus A}(1-x_j) - p_S(A) \\
			&= (k - x(A)) \prod_{j \in A \setminus i}x_j \prod_{j \in S \setminus A}(1-x_j) - x_i \prod_{j \in A \setminus i} x_j \prod_{j \in S \setminus A} (1-x_j) \\
			&= (k-x(A)) \; p_{S \setminus i}(A \setminus i) - x_i\;  p_{S \setminus i}(A \setminus i) 
			= p_{S \setminus i}(A \setminus i) \Big(k - x(A \setminus i) - 2x_i\Big).
		\end{split}
	\end{align*}
	
	\noindent For a set $A \subseteq S$ such that $i \notin A$, we get
	\begin{align*}
		\begin{split}
			\frac{\partial}{\partial x_i}p_S(A)(k-x(A)) &= - (k - x(A)) \prod_{j \in A}x_j \prod_{j \in (S \setminus A)\setminus i} (1-x_j) 
			= - p_{S \setminus i}(A)\Big(k-x(A)\Big).
		\end{split}
	\end{align*}

	\noindent We then have
	
	\begin{align*}
		\frac{\partial h_S^k(x)}{\partial x_i} &= \frac{\partial}{\partial x_i}\sum_{\substack{A \subseteq S \\ |A|=k}} p_S(A) (k-x(A)) 
		= \sum_{\substack{A \subseteq S \\ |A|=k \\ i \in A}}\frac{\partial}{\partial x_i}p_S(A) (k-x(A)) + \sum_{\substack{A \subseteq S \\ |A|=k \\ i \notin A}}\frac{\partial}{\partial x_i}p_S(A) (k-x(A)) \nonumber\\
		&= \sum_{\substack{A \subseteq S \\ |A|=k \\ i \in A}}p_{S \setminus i}(A \setminus i) \Big(k - x(A \setminus i) - 2x_i\Big) - \sum_{\substack{A \subseteq S \\ |A|=k \\ i \notin A}}p_{S \setminus i}(A)\Big(k-x(A)\Big) \nonumber\\
		&= \sum_{\substack{B \subseteq S \setminus i\\|B|=k-1}}p_{S \setminus i}(B)\Big(k - x(B) - 2x_i\Big) - \sum_{\substack{A \subseteq S \setminus i\\|A|=k}}p_{S \setminus i}(A)\Big(k-x(A)\Big) \nonumber\\
		&= \sum_{\substack{A \subseteq S \setminus i\\|A|=k-1}}p_{S \setminus i}(A)\Big(k -1 - x(A)+ 1 - 2x_i\Big) - h_{S \setminus i}^k(x) \nonumber\\
		&= \sum_{\substack{A \subseteq S \setminus i\\|A|=k-1}}p_{S \setminus i}(A)\Big(k -1 - x(A)\Big) + (1-2x_i)\sum_{\substack{A \subseteq S \setminus i\\|A|=k-1}}p_{S \setminus i}(A) - h_{S \setminus i}^k(x) \nonumber\\
		&= (1-2x_i)Q_{S \setminus i}^{k-1}(x) - \Big(h_{S \setminus i}^k(x) - h_{S \setminus i}^{k-1}(x)\Big) \nonumber \\
		&= (1-2x_i)Q_{S \setminus i}^{k-1}(x) - Q_{S \setminus i}^{k-1}(x) \Big(x(S \setminus i) - (k-1)\Big) 
		= Q_{S \setminus i}^{k-1}(x) \Big(k - x(S) - x_i\Big) 
	\end{align*}
	where we use Lemma \ref{lemma_recursive_h_S_k} in the last line. From the above it follows that $\nabla h_S^k(x) = 0$ if and only if
	\begin{equation}
		\label{eq_gradient_zero}
		Q_{S \setminus i}^{k-1}(x) \Big(k - x(S) - x_i\Big) = 0 \quad \forall i \in S.
	\end{equation}
	Notice that
	\begin{align*}
		Q_{S \setminus i}^{k-1}(x) = 0 &\iff \sum_{\substack{A \subseteq S \setminus i \\ |A|=k-1}}p_{S \setminus i}(A) = 0 \iff p_{S \setminus i}(A) = 0 \quad \forall A\subseteq S \setminus i, |A|=k-1 \\
		& \iff \prod_{j \in A}x_j \prod_{j \in (S \setminus i)\setminus A}(1-x_j) = 0 \quad \forall A \subseteq S, |A|=k-1.
	\end{align*}
	We can see this implies that such a solution lies on the boundary of $[0,1]^S$, since there exists an index $j \in S$ such that $x_j=0$ or $x_j=1$. Since we are focusing on extrema in the interior, we may disregard that solution. Hence, by \eqref{eq_gradient_zero}, we have $x_i = k - x(S)$ for all $i \in S$. All the $x_i$ are thus equal and by setting $x_i = t$ for every $i \in S$, we get
	\begin{align*}
		t = k - (n-1)t \iff t = k/n \iff x_i = k/n \quad \forall i \in S.
	\end{align*}
	Therefore, $h_S^k(x)$ has a unique extremum in the interior of $[0,1]^S$ at the point $(k/n,\dots,k/n)$.
\end{proof}

In order to prove Theorem~\ref{thm_max_hSk}, we need one additional lemma; we leave its proof to the appendix.

\begin{lemma}
	\label{lemma_inductive_step}
	The following holds for any $n \geq 2$ and $k \in \{1,\dots,n-1\}$:
	\begin{align}
		\alpha(k,n) &> \alpha(k-1,n-1), \label{lemma_inductive_first_statement}
		\\
		\alpha(k,n) &> \alpha(k,n-1). \label{lemma_inductive_second_statement}
	\end{align}
\end{lemma}


\begin{proof}[Proof of Theorem \ref{thm_max_hSk}]
	We prove the statement by induction on $n\geq 2$. 
	The base case corresponds to $n=2$ and $k=1$. In this case, we get $S = \{1\}$ and $h_S^k(x) = x_1 (1-x_1)$.
	It is easy to see that this is a parabola which attains its maximum at the point $x_1 = 1/2$ over the unit interval $[0,1]$. Moreover the function value at that point is $1/4 = \alpha(1,2)$.
	
	We now prove the induction step. Let $n \geq 3$ and $k\in \{1,\dots,n-1\}$, and assume by induction hypothesis that the statement holds for any $2 \leq n' < n$ and $k \in \{1,2,\ldots,n'-1\}$.
	
	By Proposition \ref{prop_interior}, $h_S^k(x)$ has a unique extremum 
	in the interior of $[0,1]^S$ at the point $(k/n,\dots,k/n)$. We first show that the function $h_S^k(x)$ evaluated at that point is indeed equal to $\alpha(k,n)$.
	
	\begin{align*}
		\begin{split}
			h_S^k(k/n,\dots,k/n) &= \binom{n-1}{k}\left(\frac{k}{n}\right)^k \left(1 -\frac{k}{n}\right)^{n-1-k} \left(k - k \frac{k}{n}\right) 
			= k \binom{n-1}{k}\left(\frac{k}{n}\right)^k \left(1 -\frac{k}{n}\right)^{n-k} \\
			&= k \frac{n-k}{n} \binom{n}{k}\left(\frac{k}{n}\right)^k \left(1 -\frac{k}{n}\right)^{n-k} 
			= \alpha(k,n).
		\end{split}
	\end{align*}
	
	We next show that any point on the boundary of $[0,1]^S$ has a lower function value than $\alpha (k,n)$. A point $x \in [0,1]^S$ lies on the boundary if there exists $i \in S$ such that $x_i = 0$ or $x_i=1$.
	
	\begin{itemize}
		\item Suppose there exists $i \in S$ such that $x_i = 0$. For any set $A \subseteq S$ containing $i$, we get $p_S(A) = 0$. Hence:
		\begin{align*}
			h_S^k(x) = \sum_{A \subseteq S, |A| = k}p_S(A)(k-x(A)) = \sum_{A \subseteq S\setminus i, |A| = k}p_{S \setminus i}(A)(k-x(A)) = h_{S\setminus i}^k(x).
		\end{align*}
		If $k = n-1$, then $h_{S\setminus i}^k(x) = 0$. We then clearly get $h_S^k(x) = h_{S \setminus i}^k(x) = 0 < \alpha(k,n)$.
		If $k < n-1$, then by induction hypothesis and Lemma \ref{lemma_inductive_step},
		\[h_S^k(x) = h_{S \setminus i}^k(x) \leq \alpha(k,n-1) < \alpha(k,n).\]
		
		\item Suppose there exists $i \in S$ such that $x_i = 1$. For any set $A \subseteq S$ not containing $i$, we get $p_S(A)=0$. Hence:
		\begin{align*}
			h_S^k(x) &= \sum_{A \subseteq S, |A| = k}p_S(A)(k-x(A)) = \sum_{\substack{A \subseteq S,|A|=k\\i\in A}}p_S(A)(k-x(A)) \\
			&= \sum_{\substack{A \subseteq S, |A|=k\\i\in A}}p_{S\setminus i}(A\setminus i)(k-1-x(A\setminus i)) = \sum_{\substack{A \subseteq S\setminus i\\|A|=k-1}}p_{S\setminus i}(A)(k-1-x(A)) 
			= h_{S \setminus i}^{k-1}(x).
		\end{align*}
		If $k=1$, then $h_{S \setminus i}^{k-1}(x) = 0$. We then clearly get $h_S^k(x) = h_{S \setminus i}^{k-1}(x) = 0<\alpha(k,n).$
		If $k > 1$, then by induction hypothesis and Lemma \ref{lemma_inductive_step},
		\[h_S^k(x) = h_{S \setminus i}^{k-1}(x) \leq \alpha(k-1,n-1) < \alpha(k,n).\]  
	\end{itemize}
	Since our function is continuous over a compact domain, it attains a maximum. By using continuity, together with the facts that $(k/n,\dots,k/n)$ is the unique extremum in the interior, and that it has higher function value than any point at the boundary, it follows that $(k/n,\dots,k/n)$ must be a global maximum. This completes the proof.
\end{proof}

\subsection{Proof of Theorem \ref{thm_optimal_crs_ukn}}
\label{subsection_proof_of_main_theorem_uniform}
We now have all the ingredients to prove Theorem \ref{thm_max_crs_ukn} and, therefore, Theorem \ref{thm_optimal_crs_ukn}. The two main building blocks for the proof are Lemma \ref{technical_lemma} and Theorem \ref{thm_max_hSk}. 

\begin{proof}[Proof of Theorem \ref{thm_max_crs_ukn}]
	
	By Lemma \ref{technical_lemma}, we get that for any $x \in P_{\I}$ (since $P_{\I} \subseteq \widetilde{P}_{\I}$),
	\begin{equation}
		\label{eq_first_part_of_theorem}
		G(x) \leq \sum_{A \subseteq S, \, |A| = k}p_S(A)(1-\bar{x}(A)).
	\end{equation}
	Moreover, for every $x \in P_{\I}$ satisfying $x_e = k - x(S)$, equality holds in \eqref{eq_first_part_of_theorem}.
	
	\noindent By Theorem \ref{thm_max_hSk}, we get that for any $x \in P_{\I}$,
	\begin{equation}
		\label{eq_second_part_of_theorem}
		\sum_{A \subseteq S, \, |A| = k}p_S(A)(1-\bar{x}(A)) \leq 1-c(k,n).
	\end{equation}
	Equality holds in \eqref{eq_second_part_of_theorem} if $x_i = k/n$ for every $i \in S$. This holds because the above expression does not depend on $x_e$, and the projection of the polytope $P_{\I}$ to the $S$ coordinates is included in the unit hypercube $[0,1]^S$.
	
	Therefore, by combining \eqref{eq_first_part_of_theorem} and \eqref{eq_second_part_of_theorem}, we get that $G(x) \leq 1 - c(k,n)$ for every $x \in P_{\I}$.
	Moreover, for the point $x_i = k/n$ for every $i \in N$,
	equality holds: $G(k/n,\dots,k/n) = 1 - c(k,n).$
	Indeed, \eqref{eq_first_part_of_theorem} holds with equality because $x_e = k - x(S)$ is satisfied (since $k - x(S) = k - (n-1)k/n = k/n$), and \eqref{eq_second_part_of_theorem} also holds with equality because $x_i = k/n$ for every $i \in S$.
\end{proof}

\subsection{Optimality}
\label{subsection_optimality_uniform}
In this section, we argue that a balancedness of $c(k,n)$ is in fact optimal for $U^k_n$.
Our bound is more refined than the one given in \cite{yan2011mechanism}, in the sense that it depends on both $k$ and $n$.

\begin{theorem}
	\label{thm_optimality}
	There does not exist a $c$-balanced CR scheme for the uniform matroid of rank $k$ on $n$ elements satisfying
	$c > 1 - \binom{n}{k}\:\left(1-\frac{k}{n}\right)^{n+1-k}\:\left(\frac{k}{n}\right)^k$.
\end{theorem}

\noindent The proof uses a similar argument to the one used for $U^1_n$ in \cite{crs}. It relies on computing the value $\mathbb{E}\big[r(R(x)\big]$, i.e., the expected rank of the random set $R(x)$. However, for values of $k>1$, the argument becomes more involved than the one presented in \cite{crs}. Our proof uses Lemma \ref{lemma_recursive_h_S_k}.

\begin{corollary}[of Lemma \ref{lemma_recursive_h_S_k}]
	\label{cor:eq_h_N_k_at_k_over_n}
	Let $x \in P_{\I}$ be the point $x_i = k/n$ for all $i \in N$. Then
	$h_N^k(x) = \sum_{i =0}^{k-1} Q_N^i(x)(k-i)$.
\end{corollary}

\begin{proof}[Proof of Theorem \ref{thm_optimality}]
	Let $\pi$ be an arbitrary $c$-balanced CR scheme for $U^k_n$, and fix the point $x$ such that $x_i = \frac{k}{n}$ for every $i \in N.$
	Clearly, $x \in P_{\I} = \{x \in [0,1]^N: x_1 + \dots + x_n \leq k\}$.
	Let $R(x)$ be the random set satisfying $\mathbb{P}[i \in R(x)] = x_i$ for each $i$ independently, and denote by $I := \pi_x(R(x))$ the set returned by the CR scheme $\pi$.
	By definition of a CR scheme, we have $\mathbb{E}[|I|] \leq \mathbb{E}\big[r(R(x))\big]$ and 
	$\mathbb{E}[|I|] = \sum_{i \in N}\mathbb{P}[i \in I] \geq \sum_{i \in N} c \: x_i = \frac{n c k}{n} = c k$.
	It follows that
	$c \leq \mathbb{E}\big[r(R(x))\big] / k$.
	Moreover, recall that
	\begin{equation}
		\label{eq_proba_Rx_equal_i}
		\mathbb{P}[|R(x)| = i] = \sum_{A \subseteq N, \, |A|=i}p_N(A) = Q_N^i(x).
	\end{equation}
	We then have
	\begin{align*}
		\mathbb{E}[r(R(x))] &= \sum_{i=0}^k i \; \mathbb{P}[r(R(x)) = i] = \sum_{i=0}^{k-1} i \; \mathbb{P}\big[|R(x)| = i\big] + k \; \mathbb{P}\big[|R(x)| \geq k\big]\\
		&= \sum_{i=0}^{k-1} i \; \mathbb{P}\big[|R(x)| = i\big] + k \; \Big(1 - \mathbb{P}\big[|R(x)| \leq k-1\big]\Big)\\
		&= k + \sum_{i=0}^{k-1} i \; \mathbb{P}\big[|R(x)| = i\big] - k \sum_{i=0}^{k-1} \; \mathbb{P}\big[|R(x)| = i\big] 
		= k - \sum_{i=0}^{k-1} (k-i) Q_N^i(x)
		= k - h_N^k(x) \\
		&= k - \sum_{A \subseteq N, |A| = k}p_N(A)(k - x(A)) 
		= k - \sum_{A \subseteq N, |A| = k} \left(\frac{k}{n}\right)^k \left(1-\frac{k}{n}\right)^{n-k} \left(k - k \frac{k}{n}\right)\\
		&= k \left(1 - \binom{n}{k}\:\left(1-\frac{k}{n}\right)^{n+1-k}\:\left(\frac{k}{n}\right)^k\right),
	\end{align*}
	where the last two equalities in the third line follow by \eqref{eq_proba_Rx_equal_i} and Corollary~\ref{cor:eq_h_N_k_at_k_over_n} respectively.
	Combining this with the bound $c \leq \mathbb{E}\big[r(R(x))\big] / k$ leads to the desired result.
\end{proof}

\subsection{Marginal viewpoint of CR schemes and efficient implementation}
\label{sec:efficient}

An important object in the design of CR schemes are the marginals.
Given a CR scheme $\pi$, a vector $x \in P_\mathcal{I}$, and a set $A \subseteq supp(x)$, the CR scheme returns a (potentially random) set $\pi_x (A) \in \mathcal{I}$ with $\pi_x (A) \subseteq A$. This defines a probability distribution over subsets of $A$, and hence also a set of marginals. More precisely, the marginals of $\pi_x(A)$ are given by $y_x^A \in [0,1]^N$ where  $(y_x^A)_e := \mathbb{P}[e \in \pi_x(A)]$. It is immediate that $y_x^A \in P_\mathcal{I}$, since $\pi_x(A) \in \mathcal{I}$.

Marginals are heavily used in \cite{crs_bipartite_matchings} to design an optimal CR scheme for bipartite matchings. Our next result provides an explicit formula for the marginals of the CR scheme described in Algorithm~\ref{algo_optimal_crs_ukn}. We postpone its proof to the appendix.

\begin{lemma}
	\label{lemma_prob_e_in_pi_x_A}
	The marginals of the CR scheme defined in Algorithm \ref{algo_optimal_crs_ukn} satisfy:
	\[
	\mathbb{P}[e \in \pi_x(A)] = \frac{k - x_e}{|A|} + \frac{x(A \setminus e)}{|A|(|A|-1)}.
	\]
\end{lemma}

We now discuss an alternative viewpoint of the scheme provided in Algorithm~\ref{algo_optimal_crs_ukn}, which yields a more efficient implementation and a polynomial time algorithm for any value of $k$.
We use the marginal viewpoint of CR schemes described in \cite[see Section 2 and Proposition 1]{crs_bipartite_matchings}. 
As pointed out there, while the marginals carry less information than the CR scheme, the balancedness and monotonicity properties can be determined solely from the marginals --- see Definition~\ref{def:crs}. In particular, if two CR schemes have the same set of marginals, then they also have the same balancedness, and one scheme is monotone if and only if the other one is.

Lemma~\ref{lemma_prob_e_in_pi_x_A} gives an explicit formula for the marginals $y_x^A$ of the CR scheme defined in Algorithm~\ref{algo_optimal_crs_ukn}. Using this, one can design an efficient CR scheme with the same marginals as follows: First decompose $y_x^A \in P_\mathcal{I}$ into a convex combination of vertices of the polytope: $y_x^A = \sum_{i \in [m]} \lambda_i \mathbf{1}_{A_i}$, where $A_i \in \mathcal{I}$ for $i \in [m]$. Then, output the set $A_i$ with probability $\lambda_i$. Note that since the $\lambda_i$'s are a convex combination, the above sampling procedure is a well-defined probability distribution. Moreover, the convex combination $y_x^A = \sum_{i \in [m]} \lambda_i \mathbf{1}_{A_i}$ can be computed efficiently via standard methods --- see for instance \cite[Corollary 40.4a]{schrijver2003combinatorial} or \cite[Corollary 14.1f and 14.1g]{schrijver1998theory}. 

\subsection{Monotonicity}
\label{sec:monotonicity}
 
We next argue that Algorithm \ref{algo_optimal_crs_ukn} is a monotone CR scheme. This is a desirable property for CR schemes, since they can then be used to derive approximation guarantees for constrained submodular maximization problems. We need Lemma \ref{lemma_prob_e_in_pi_x_A}, i.e., the marginals of the CR scheme.

\begin{theorem}
	\label{thm_monotonicity_cr_scheme_ukn}
	Algorithm \ref{algo_optimal_crs_ukn} is a monotone CR scheme for $U^k_n$. That is, for every $x \in P_{\I}$ and $e \in A \subseteq B \subseteq \text{supp}(x)$, we have $\mathbb{P}[e \in \pi_x(A)] \geq \mathbb{P}[e \in \pi_x(B)].$
\end{theorem}
\begin{proof}
	Let $A \subseteq \text{supp}(x)$ and $e \in A$. If $|A| \leq k$, then $\mathbb{P}[e \in \pi_x(A)] = 1$,
	and the theorem trivially holds. We therefore suppose that $|A| > k$. In order to prove the theorem, it is clearly enough to show that for any $f \in \text{supp}(x) \setminus A$,
	\begin{equation}
		\label{eq_prob_e_in_pi_x_A_greater_than}
		\mathbb{P}[e \in \pi_x(A)] \geq \mathbb{P}[e \in \pi_x(A \cup \{f\})].
	\end{equation}
	We show the difference of those two terms is greater than 0 by using Lemma \ref{lemma_prob_e_in_pi_x_A} for both terms:
	
	\begin{align*}
		\mathbb{P}[e &\in \pi_x(A)] - \mathbb{P}[e \in \pi_x(A \cup \{f\})] = \frac{k-x_e}{|A|} + \frac{x(A\setminus e)}{|A|(|A|-1)} - \frac{k-x_e}{|A|+1} - \frac{x(A\setminus e)+x_f}{(|A|+1)|A|} \\
		&= \frac{k-x_e}{|A|} - \frac{k-x_e}{|A|+1} - \frac{x_f}{(|A|+1)|A|} + x(A \setminus e) \left(\frac{1}{|A|(|A|-1)} - \frac{1}{(|A|+1)|A|}\right)\\
		&= \frac{(|A|+1)(k-x_e) - |A|(k-x_e) - x_f}{|A|(|A|+1)} + \frac{2x(A\setminus e)}{(|A|^2-1)|A|}\\
		&= \frac{k - x_e - x_f}{|A|(|A|+1)} + \frac{2x(A\setminus e)}{(|A|^2-1)|A|} \geq 0.
	\end{align*}
	The last inequality holds because since $x \in P_{\I}= \{x \in [0,1]^N \mid x(N) \leq k\}$, we have $x_e + x_f \leq k$, and all the other terms are positive. We have thus shown \eqref{eq_prob_e_in_pi_x_A_greater_than} which is enough to prove the theorem.
\end{proof}

\subsection{Extension to partition matroids}
\label{sec:partition}

A CR scheme for uniform matroids can be naturally extended to a CR scheme for \emph{partition matroids}. This is not surprising since partition matroids can be seen as a direct sum of uniform matroids --- see Example~\ref{ex:partition}. For completeness, in this section we discuss how the results from Section~\ref{sec:uniform} lead to an optimal CR scheme for partition matroids. This is encapsulated in the following two results.

\begin{proposition}
	Let $\mathcal{M} = (N, \I)$ be a partition matroid given by $\I = \{A \subseteq N: |A \cap D_i| \leq d_i, \; \forall i \in \{1,\dots,k\} \}$.
	If there is a (monotone) $\alpha(k,n)$-balanced CR scheme for the uniform matroid $U_n^k$, then there is a (monotone) $\alpha$-balanced CR scheme for $\mathcal{M}$, where $\alpha=\min_{i \in [k]} \alpha(d_i,|D_i|)$.
\end{proposition}
\begin{proof}
	For each $i \in [k]$, let $\pi^i$ be an $\alpha(d_i, |D_i|)$-balanced CR scheme for the uniform matroid $U^{d_i}_{|D_i|}$. 
	Let $P_i$ denote the matroid polytope of the uniform matroid $U^{d_i}_{|D_i|}$, and $P_{\I}$ denote the matroid polytope of the partition matroid $\mathcal{M} = (N, \I)$.
	Given any $x \in P_{\I}$, let $x^i \in [0,1]^{D_i}$ denote the restriction of $x$ to $D_i$. Since $x \in P_{\I}$, it is clear that $x^i \in P_i$.
	
	Consider the CR scheme $\pi$ defined as follows, $\pi_x(A) = \biguplus_{i \in [k]} \pi_{x^i}^i(A \cap D_i)$. 
	That is, we run the CR schemes $\pi^i$ independently in each partition $D_i$ and take the (disjoint) union of their outputs.
	Let $e \in N$ be such that $e \in D_i$. Then,
	$\mathbb{P}[e \in \pi_x(R(x))] = \mathbb{P}[e \in \pi^i_{x^i}(R(x) \cap D_i)]  \geq \alpha(d_i,|D_i|) \cdot x_e$.
	Hence, it follows that $\pi$ is an $\alpha$-balanced CR scheme for $\mathcal{M}$, where $\alpha=\min_{i \in [k]} \alpha(d_i,|D_i|)$.
	
	For the monotonicity part, assume that the CR schemes $\pi^i$ defined above are all monotone. 
	Let $x \in P_{\I}$,  $e \in A \subseteq B$, and $i \in [k]$ be the unique index such that $e \in D_i$. Then, 
	$\mathbb{P}[e \in \pi_x(A)] = \mathbb{P}[e \in \pi^i_{x_i}(A \cap D_i)] \geq \mathbb{P}[e \in \pi^i_{x_i}(B \cap D_i)] = \mathbb{P}[e \in \pi_x(B)].$
	Hence $\pi$ is also monotone. 
\end{proof}

\begin{proposition}
	Let $\mathcal{M} = (N, \I)$ be a partition matroid given by $\I = \{A \subseteq N: |A \cap D_i| \leq d_i, \; \forall i \in \{1,\dots,k\} \}$.
	If there is no $\alpha(k,n)$-balanced CR scheme for the uniform matroid $U_n^k$, then there is no $\alpha$-balanced CR scheme for $\mathcal{M}$, where $\alpha=\min_{i \in [k]} \alpha(d_i,|D_i|)$.
\end{proposition}
\begin{proof}
	Assume such a CR scheme $\pi$ exists, and let $j \in \text{argmin}_{i \in \{1, \dots, k \}} \alpha(d_i, |D_i|)$. 
	Let $P$ denote the matroid polytope of the uniform matroid $U^{d_j}_{|D_j|}$, and let $P_{\I}$ denote the matroid polytope of the partition matroid $\mathcal{M} = (N, \I)$.
	Then, for any $\bar{x} \in P$, let $x \in [0,1]^N$ be defined as $x_e = \bar{x}_e$ if $e \in D_j$, and $x_e = 0$ otherwise. Clearly, $x \in P_{\I}$ since $\bar{x} \in P$. Hence, $\mathbb{P}[e \in \pi_x(R(x)) \mid e \in R(x)] \geq \alpha = \alpha(d_j,|D_j|)$. But this contradicts the assumption that there is no $\alpha(d_j,|D_j|)$-balanced CR scheme for the uniform matroid $U^{d_j}_{|D_j|}$.
\end{proof}

\section{Conclusion}

Contention resolution schemes are a general and powerful tool for rounding a fractional point in a relaxation polytope. It is known that matroids admit $(1-1/e)$-balanced CR schemes, and that this is the best possible. This impossibility result is in particular true for uniform matroids of rank one. 
For uniform matroids of rank $k$ (i.e., cardinality constraints), one can get a $(1 - e^{-k} k^k/k!)$-balanced CR scheme by combining a reduction from \cite{crs} and a result from \cite{yan2011mechanism}. The main drawback of this approach, however, is its lack of simplicity.
In this work, we provide an explicit and much simpler scheme with a balancedness of $c(k,n):= 1 - \binom{n}{k}\left(1-\frac{k}{n}\right)^{n+1-k}\left(\frac{k}{n}\right)^k$. In particular, $c(k,n) > 1 - e^{-k} k^k/k!$ for every $n$, and $c(k,n)$ converges to $1 - e^{-k} k^k/k!$ as $n$ goes to infinity. Our balancedness is therefore better for every fixed $n$, and achieves $1 - e^{-k} k^k/k!$ asymptotically. We also show optimality and monotonicity of our scheme, and discuss how it naturally extends to an optimal CR scheme for partition matroids.

We believe that finding other classes of matroids where the $1-1/e$ balancedness factor can be improved is an interesting direction for future work.
Moreover, while this work focused on the offline setting, this question can also be studied in the context where the elements of $R(x)$ arrive in an online fashion (e.g., in the case of random or online contention resolution schemes). Finally, it would also be interesting to see whether a simpler proof of the optimality of our algorithm can be obtained.

\subsubsection*{Acknowledgements}
We thank Chandra Chekuri and Vasilis Livanos for useful feedback and for mentioning the work of \cite{barman2021tight}. We also thank the anonymous reviewers for valuable suggestions, and for pointing out the connection with the work of \cite{yan2011mechanism}.


\bibliographystyle{plain}
\bibliography{references}

\appendix
\section{Appendix}

  \subsection*{Proof of Lemma \ref{lemma_recursive_h_S_k}}
  \begin{proof}[Proof of Lemma \ref{lemma_recursive_h_S_k}]
  	Notice that for $i \in A$, $p_S(A)\:(1-x_i) = p_S(A \setminus i) \: x_i$. Then

  	\begin{align}
  		h_S^k(x) &= \sum_{\substack{A \subseteq S\\ |A|=k}} p_S(A)\:\sum_{i \in A} (1-x_i) = \sum_{\substack{A \subseteq S\\ |A|=k}} \:\sum_{i \in S} \: p_S(A) (1-x_i) \mathbf{1}_{\{i \in A\}} \nonumber\\
  		&= \sum_{i \in S} \sum_{\substack{A \subseteq S\\ |A|=k}} x_i \: p_S(A \setminus i)\: \mathbf{1}_{\{i \in A\}} = \sum_{i \in S} \sum_{\substack{B \subseteq S \setminus i\\|B|=k-1}} x_i \; p_S(B) \nonumber \\
  		&= \sum_{i \in S} x_i \sum_{\substack{A \subseteq S \setminus i\\|A|=k-1}} p_S(A) = \sum_{i \in S} x_i \left(\sum_{\substack{A \subseteq S \\|A|=k-1}} p_S(A) - \sum_{\substack{A \subseteq S \nonumber \\|A|=k-1}} p_S(A) \mathbf{1}_{\{i \in A\}}\right) \\
  		&= x(S) Q_S^{k-1}(x) - \sum_{\substack{A \subseteq S \\|A|=k-1}}p_S(A)x(A). \label{eq_h_N_part_one}
  	\end{align}
  	Notice that by definition of $h^k_S(x)$ we have
  	\begin{align}
  		h_S^{k-1}(x) &= \sum_{\substack{A \subseteq S\\|A|=k-1}}p_S(A)(k-1-x(A)) = (k-1)\sum_{\substack{A \subseteq S\\|A|=k-1}}p_S(A) - \sum_{\substack{A \subseteq S \\|A|=k-1}}p_S(A)x(A) \nonumber \\&= (k-1) \: Q_S^{k-1}(x) - \sum_{\substack{A \subseteq S \\|A|=k-1}}p_S(A)x(A). \label{eq_h_N_Q_N}
  	\end{align}
  	Substracting \eqref{eq_h_N_Q_N} from \eqref{eq_h_N_part_one} we get
  	\[
  		h_S^k(x) - h_S^{k-1}(x) = Q_S^{k-1}(x) \Big(x(S) - (k-1) \Big).
  	\]
  	We can rewrite this recursive formula as
	\[
  		h_S^{i+1}(x) - h_S^i(x) = Q_S^i(x) \Big(x(S)-i \Big).
	\]
  	By summing both sides from $0$ to $k-1$ and noticing that $h_S^0(x) = 0$, we get the desired result.
  \end{proof}

  \subsection*{Proof of Lemma \ref{lemma_inductive_step}}
  \begin{proof}[Proof of Lemma \ref{lemma_inductive_step}]
  	First, notice that the function $g(x):=\left(\frac{x-1}{x}\right)^x$ is strictly increasing for $x \geq 1$.
  	Indeed, by using the strict inequality $\log(1+x) < x$ for any $x>0$, we see that the derivative of $\log(g(x))$ is strictly positive:
  	
  	\begin{align*}
  		\frac{d}{dx} \log(g(x)) &= \frac{d}{dx} x \log\left(\frac{x-1}{x}\right) = \log\left(\frac{x-1}{x}\right) + x \frac{x}{x-1} \frac{1}{x^2} = \log\left(\frac{x-1}{x}\right) + \frac{1}{x-1}\\
  		&= \frac{1}{x-1} - \log\left(\frac{x}{x-1}\right) = \frac{1}{x-1} - \log\left(1 + \frac{1}{x-1}\right) > 0.
  	\end{align*}
  	
  	\noindent We first prove \eqref{lemma_inductive_first_statement}. If $k = 1$, then $\alpha(k-1,n-1) = 0$ and the statement clearly holds. We may thus assume $k>1$. Then,
  	
  	\begin{align*}
  		\frac{\alpha(k,n)}{\alpha(k-1,n-1)} &= \frac{k}{k-1} \; \frac{\binom{n}{k}}{\binom{n-1}{k-1}} \; \left(\frac{n-k}{n}\right)^{n+1-k}\; \left(\frac{k}{n}\right)^k \left(\frac{n-1}{n-k}\right)^{n+1-k}\; \left(\frac{n-1}{k-1}\right)^{k-1} \\
  		&= \frac{k}{k-1} \; \frac{n}{k} \; \frac{k^k \; (n-1)^n}{n^{n+1}(k-1)^{k-1}}
  		= \left(\frac{n-1}{n}\right)^n \left(\frac{k}{k-1}\right)^{k}
  		= \frac{g(n)}{g(k)} > 1.
  	\end{align*}
  	
  	\noindent We now prove \eqref{lemma_inductive_second_statement}. If $k=n-1$, then $\alpha(k,n-1) = 0$ and the statement clearly holds. We may thus assume $k < n-1$. Then,
  	
  	\begin{align*}
  		\frac{\alpha(k,n)}{\alpha(k,n-1)} &= \frac{\binom{n}{k}}{\binom{n-1}{k}} \; \left(\frac{n-k}{n}\right)^{n+1-k}\; \left(\frac{k}{n}\right)^k \left(\frac{n-1}{n-1-k}\right)^{n-k}\; \left(\frac{n-1}{k}\right)^{k}\\
  		&= \frac{n}{n-k} \frac{(n-k)^{n+1-k} \; (n-1)^n}{n^{n+1} \; (n-1-k)^{n-k}}
  		= \left(\frac{n-1}{n}\right)^n \left(\frac{n-k}{n-k-1}\right)^{n-k}
  		= \frac{g(n)}{g(n-k)} > 1.\qedhere
  	\end{align*}
  \end{proof}

   \subsection*{Proof of Lemma \ref{lemma_prob_e_in_pi_x_A}}
  \begin{proof}[Proof of Lemma \ref{lemma_prob_e_in_pi_x_A}]
  	\begin{align}
  		\mathbb{P}[e \in \pi_x(A)] &= \sum_{\substack{B \subseteq A\\|B|=k\\e \in B}}q_A(B) = \sum_{\substack{B \subseteq A \setminus e \\ |B|=k-1}}q_A(B \cup e)\nonumber\\
  		&= \sum_{\substack{B \subseteq A \setminus e \\ |B|=k-1}}\frac{1}{\mbinom{|A|}{k}}\left(1 + \frac{x(A \setminus e) - x(B)}{|A|-k} - \frac{x(B)+x_e}{k}\right)\nonumber\\
  		&= \sum_{\substack{B \subseteq A \setminus e \\ |B|=k-1}}\frac{1}{\mbinom{|A|}{k}}\left(1 - \frac{x_e}{k} + \frac{x(A \setminus e)}{|A|-k} - x(B)\left(\frac{1}{|A|-k} + \frac{1}{k}\right)\right)\nonumber\\
  		&= \sum_{\substack{B \subseteq A \setminus e \\ |B|=k-1}}\frac{1}{\mbinom{|A|}{k}}\left(\frac{k-x_e}{k} + \frac{x(A \setminus e)}{|A|-k} - x(B)\frac{|A|}{k(|A|-k)} 
  		\right). \label{eq_prob_e_in_pi_x_A_first_part}
  	\end{align}
  	We now use Equation \eqref{eq_sum_x_B} in a slightly modified form:
  	\begin{equation}
  		\label{eq_sum_x_b_slightly_modified}
  		\sum_{\substack{B \subseteq A \setminus e \\ |B|=k-1}}x(B) = \mbinom{|A|-2}{k-2}x(A\setminus e).
  	\end{equation}
  	The only part in the sum \eqref{eq_prob_e_in_pi_x_A_first_part} that depends on $B$ is the last term with $x(B)$. Hence, by plugging-in \eqref{eq_sum_x_b_slightly_modified} into \eqref{eq_prob_e_in_pi_x_A_first_part}, we get:
  	\begin{equation}
  		\label{eq_prob_e_in_pi_x_A_second_part}
  		\mbinom{|A|}{k}\mathbb{P}[e \in \pi_x(A)] = \mbinom{|A|-1}{k-1}\ \left(\frac{k-x_e}{k}\right) + \mbinom{|A|-1}{k-1}\frac{x(A \setminus e)}{|A|-k} - \mbinom{|A|-2}{k-2}\frac{|A|}{k(|A|-k)}x(A \setminus e).
  	\end{equation}
  	We now use the formula
  	$\binom{n}{k} = \frac{n}{k}\binom{n-1}{k-1}
  	$
  	to remove all the binomial coefficients from \eqref{eq_prob_e_in_pi_x_A_second_part}. We get
  	\begin{align*}
  		\frac{|A|}{k} \mathbb{P}[e \in \pi_x(A)] &= \frac{k-x_e}{k} + \frac{x(A\setminus e)}{|A|-k} - \frac{k-1}{|A|-1}\; \frac{|A|}{k(|A|-k)}x(A \setminus e) \\
  		&= \frac{k-x_e}{k} + \frac{x(A\setminus e)}{|A|-k}\left(1 - \frac{|A|(k-1)}{(|A|-1)k}\right) 
  		= \frac{k-x_e}{k} + \frac{x(A\setminus e)}{k(|A|-1)}.
  	\end{align*}
  	This implies the desired result:
  	\[
  	\mathbb{P}[e \in \pi_x(A)] = \frac{k-x_e}{|A|} + \frac{x(A\setminus e)}{|A|(|A|-1)}. \qedhere
  	\]
  \end{proof}

  \subsection*{Proof of local maximality in Proposition \ref{prop_interior}}
  \begin{proof}
  	
  	We want to show that the point $(k/n,\dots,k/n)$ is a local maximum. We do that by computing the Hessian matrix $H(x)$ and showing that $H(k/n,\dots,k/n)$ is negative definite. Note that $H(x)$ is a $(n-1)\times(n-1)$ matrix defined by:
  	
  	\begin{align*}
  		\label{eq_hessian_matrix}
  		H(x)_{i,j}= \frac{\partial^2h_S^k(x)}{\partial x_i \partial x_j}.
  	\end{align*}
  	By some simple computations, we have
  	\begin{equation}
  		\frac{\partial^2 h_S^k}{\partial x_i^2}(k/n,\dots,k/n) = -2 \binom{n-2}{k-1}\left(\frac{k}{n}\right)^{k-1}\left(\frac{n-k}{n}\right)^{n-k-1} \qquad \forall i \in S
  	\end{equation}
  	and
  	\begin{equation}
  		\frac{\partial^2h_S^k}{\partial x_i \partial x_j}(k/n,\dots,k/n) = -\binom{n-2}{k-1}\left(\frac{k}{n}\right)^{k-1}\left(\frac{n-k}{n}\right)^{n-k-1} \qquad \text{for } i \neq j
  	\end{equation}
  	
  	\noindent Therefore,
  	
  	\begin{equation}
  		\label{eq_hessian_at_k_over_n}
  		H(k/n,\dots,k/n) = - c
  		\begin{pmatrix}
  			2 & 1 & 1 & \dots & 1 \\
  			1 & 2 & 1 & \dots & 1 \\
  			\vdots & & & & \vdots \\
  			1 & 1 & 1 & \dots & 2
  		\end{pmatrix}
  		=: -c \: A
  	\end{equation}
  	where \[c:= \binom{n-2}{k-1}\left(\frac{k}{n}\right)^{k-1}\left(\frac{n-k}{n}\right)^{n-k-1}>0.\]
  	
  	Our goal is to show that $H(k/n, \dots, k/n) \in \mathbb{R}^{(n-1) \times (n-1)}$ is negative-definite. Notice that $\lambda$ is an eigenvalue of $H(k/n,\dots,k/n)$ with corresponding eigenvector $v \in \mathbb{R}^{n-1}$ if and only if $-\lambda / c$ is an eigenvalue of $A$ with the same eigenvector $v \in \mathbb{R}^{n-1}$. It is thus enough to show that $A$ is positive-definite, i.e., all the eigenvalues of $A$ are positive.
  	
  	Notice that $A = I_{n-1} + J_{n-1}$, where $I_{n-1}$ and $J_{n-1}$ are respectively the identity matrix and the all-ones matrix of size $(n-1)\times(n-1)$. In particular, we may rewrite this as 
  	\begin{equation}
  		A = I_{n-1} + e \: e^T
  	\end{equation}
  	where $e\in \mathbb{R}^{n-1}$ is the all-ones vector.
  	
  	Let $\mu$ be an eigenvalue of $A$ with corresponding eigenvector $v$. Then
  	
  	\begin{align*}
  		Av = \mu v &\iff v + (e^Tv)e = \mu v \nonumber\\
  		&\iff  (e^Tv)e = (\mu - 1)v.
  	\end{align*}
  	
  	\begin{itemize}
  		\item If $\mu = 1$, the corresponding eigenspace is $\{v \in \mathbb{R}^{n-1} \mid e^T v = 0\}$. This eigenspace is a hyperplane of dimension $n-2$, which means that there exists $n-2$ linearly independent eigenvectors corresponding to the eigenvalue $\mu = 1$.
  		\item If $\mu \neq 1$, then we see that $e$ and $v$ are collinear, which means that $e$ is an eigenvector corresponding to $\mu$. We compute the value of $\mu$:
  		\begin{align*}
  			Ae=\mu e \iff e + (e^Te)e = \mu e \iff e +(n-1)e = \mu e \iff \mu = n.
  		\end{align*}
  	\end{itemize}
  	
  	\noindent Hence, the spectrum of $A$ is equal to $\{1,n\}$, where the multiplicity of the eigenvalue $1$ is $n-2$, whereas the multiplicity of the eigenvalue $n$ is $1$. We have therefore just proven that $A$ is positive-definite, which, by \eqref{eq_hessian_at_k_over_n}, implies that $H(k/n,\dots,k/n)$ is negative-definite and concludes the proof.
  \end{proof}

\end{document}